\documentclass[11pt]{article}
\usepackage{fullpage}
\usepackage{times}
\usepackage{amsmath,amsfonts,amssymb,amsthm}
\usepackage{enumerate}
\usepackage{xcolor}
\usepackage{tikz}
\usetikzlibrary{decorations.pathreplacing}

\newtheorem{remark}{Remark}
\newtheorem{theorem}{Theorem}
\newtheorem{lemma}{Lemma}

\newtheorem{proposition}{Proposition}

\newtheoremstyle{restate}{}{}{\itshape}{}{\bfseries}{~(restated).}{.5em}{\thmnote{#3}}
\theoremstyle{restate}
\newtheorem*{restate}{}

\newcommand{\cP}{\mathcal{P}}
\newcommand{\cU}{\mathcal{U}}
\newcommand{\cQ}{\mathcal{Q}}
\newcommand{\cD}{\mathcal{D}}

\newcommand{\cA}{\mathcal{A}}
\newcommand{\bF}{\mathbb{F}}
\newcommand{\bZ}{\mathbb{Z}}
\newcommand{\bud}[1]{Batch-Update$^{\cD}_{#1}$}
\newcommand{\buf}{{\it buf}}
\newcommand{\batch}{\textrm{bu}}
\newcommand{\lin}{\textrm{lin-skt}}
\newcommand{\Count}{\texttt{Count}}
\newcommand{\CountM}{\texttt{CountMin}}
\newcommand{\WR}{word~RAM}
\newcommand{\WRMM}{word~RAM$^{\text{MM}}$}

\def \eps {{\varepsilon}}
\def \F {{\mathbb F}}
\def \N {{\mathbb N}}
\def \Z {{\mathbb Z}}
\def \poly {{\mathrm{poly}}}

\begin{document}
\title{Faster Update Time for Turnstile Streaming Algorithms}
\author{Josh Alman\thanks{Harvard University. \texttt{joshuaalman@gmail.com}. Supported in part by NSF CCF-1651838 and NSF CCF-1741615.} \and Huacheng Yu\thanks{Princeton University. \texttt{yuhch123@gmail.com}. Supported in part by ONR grant N00014-17-1-2127, a Simons Investigator Award and NSF Award CCF 1715187.}}
\date{\today}

\setcounter{page}{0}
\maketitle
\thispagestyle{empty}

\begin{abstract}
In this paper, we present a new algorithm for maintaining linear sketches in turnstile streams with faster update time.
As an application, we show that $\log n$ \Count{} sketches or \CountM{} sketches with a constant number of columns (i.e., buckets) can be implicitly maintained in \emph{worst-case} $O(\log^{0.582} n)$ update time using $O(\log n)$ words of space, on a standard word RAM with word-size $w=\Theta(\log n)$.
The exponent $0.582\approx 2\omega/3-1$, where $\omega$ is the current matrix multiplication exponent.
Due to the numerous applications of linear sketches, our algorithm improves the update time for many streaming problems in turnstile streams, in the high success probability setting, without using more space, including $\ell_2$ norm estimation, $\ell_2$ heavy hitters, point query with $\ell_1$ or $\ell_2$ error, etc.
Our algorithm generalizes, with the same update time and space, to maintaining $\log n$ linear sketches, where each sketch
\begin{enumerate}
	\item partitions the coordinates into $k<\log^{o(1)} n$ buckets using a $c$-wise independent hash function for constant $c$,
	\item maintains the sum of coordinates for each bucket.
\end{enumerate}
Moreover, if arbitrary word operations are allowed, the update time can be further improved to $O(\log^{0.187} n)$, where $0.187\approx \omega/2-1$.
Our update algorithm is adaptive, and it circumvents the non-adaptive cell-probe lower bounds for turnstile streaming algorithms by Larsen, Nelson and Nguy{\^{e}}n (STOC'15).

On the other hand, our result also shows that proving unconditional cell-probe lower bound for the update time seems very difficult, even if the space is restricted to be (nearly) the optimum.
If $\omega=2$, the cell-probe update time of our algorithm would be $\log^{o(1)} n$.
Hence, proving any higher lower bound would imply $\omega>2$.

\end{abstract}

\newpage
\section{Introduction}
Linear sketching has numerous applications in streaming algorithms (to list a few~\cite{AMS99,CCF04,Indyk06,Li08,KNW10}), especially for \emph{turnstile} streams~\cite{Mut05}.
In the turnstile streaming model, the data structure wants to maintain a vector $\nu\in \bZ^n$, under updates of form $(u, \Delta)$ for $u\in [n]$ and $\Delta\in\bZ$, which changes $\nu_u$ to $\nu_u+\Delta$.
Each entry $\nu_i$'s is usually assumed to be bounded by $\mathrm{poly}\, n$ at any time.
Occasionally, the data structure must answer queries about $\nu$.
Special cases of turnstile streams include the incremental streams ($\Delta=1$) and singleton insertions and deletions ($\Delta=\pm 1$).
For turnstile streams, a linear-sketching streaming algorithm maintains a vector $A\nu$ in memory, for some random (but carefully sampled) matrix $A\in \bZ^{s\times n}$ and $s\ll n$, such that the (much shorter) vector $A\nu$ provides sufficient information to answer the queries with good probability.
Due to the nature of linear transformations, the updates can be maintained easily.
To handle an update $(u,\Delta)$, the algorithm simply computes $\Delta\cdot Le_u$, and adds it to $A\nu$, where $e_u$ is the $u$-th unit vector.
% These algorithms usually use small space, and have efficient \emph{update time}, i.e., how much time the algorithm uses to handle each update.

However, one single instance of linear sketch is often not very reliable, e.g., one basic \Count{} sketch~\cite{AMS99,CCF04} could only provide a constant approximation of $x$'s second moment ($\|x\|_2^2$) with a (small) constant probability.
The standard approach to boost the reliability, say to $1/n^2$ failure probability, is to maintain $O(\log n)$ independent sketches, i.e., independently sample $O(\log n)$ $A$ (or equivalently, sample a taller $A$ with $\log n$ times more rows).
This approach would naturally cause a blowup in the space by a factor of $\log n$, as well as a slowdown in the update time.
It was shown by Jayram and Woodruff~\cite{JW13} that the space blowup is necessary, however, it is unclear if one also needs to pay the $\log n$ factor in the update time.
The update time could sometimes be even \emph{more important} than the space~\cite{LNN15}, since the updates can arrive from the data stream at an extremely high rate (e.g., see~\cite{TZ12}).
A single instance of the linear sketch usually occupies logarithmic or less space, an extra $\log n$ factor is often acceptable.
However, if the update processing speed does not at least match the rate at which the updates in the stream arrive, the whole system might fail.

In linear-sketch based algorithms, the space is proportional to the number of rows in $A$, and the update time is proportional to the \emph{column sparsity}, which is the number of non-zero entries in each column (since we only need to update the non-zero entries in $\Delta \cdot Ae_u$).
% For the aforementioned problems with space lower bounds, the number of rows must increase by a $\log n$ factor.
Hence, one conceivable and common approach to speed up the update time is to use a more sparse $A$.
% One may hope to design a sparse $A$ with such number of rows, and in principle, the column sparsity does not even have to increase.
Unfortunately, this approach is known to have limitations.
Larsen, Nelson and Nguy{\^{e}}n~\cite{LNN15} proved that $A$ cannot have few rows and be column sparse at the same time for several turnstile streaming problems.
%, if we want the failure probability to be small (see the next subsection).
In fact, they showed tradeoffs between update time and space for the more general \emph{non-adaptive} data structures.
That is, the memory locations written or read during each update, can only depend on the updated index $u$ and the random bits (but not the memory contents of the data structure).
Note that all linear-sketch based algorithms are non-adaptive.
Their lower bounds also apply to cell-probe data structures, i.e., the above lower bounds hold even if we only count the number of memory words read or written by the data structure.
For the problems considered in~\cite{LNN15}, all known solutions use non-adaptive update algorithms.
Hence, in order to obtain faster update times, we would require a completely new strategy.

% The main problem we discuss in this paper is:
% \begin{center}
% 	How can we efficiently maintain $Ax$, for a possibly dense matrix $A$?
% \end{center}

\subsection{Our results}
We propose a generic solution for efficiently maintaining linear sketches using fast matrix multiplication, and apply it to a large class of streaming problems, including the problems discussed in~\cite{LNN15}.
Our new algorithm is adaptive, and hence, it circumvents the previous non-adaptive lower bound.

Our algorithm applies to any linear sketches of the following form:
\begin{enumerate}
	\setlength{\itemsep}{-3pt}
	\item partition all coordinates $[n]$ into $k$ buckets using a $c$-wise independent hash function $h$ for some constant $c$;\footnote{Most streaming algorithms need no more than constant-wise independence.}
	\item maintain the sum $\sum_{u:h(u)=b} \nu_u$ for each bucket $b\in[k]$ (equivalently, so far $A$ has $k$ rows, and each column has exactly one $1$ in a $c$-wise independently chosen row);
	\item repeat $T$ times independently.
\end{enumerate}
Denote such a linear sketch by $\lin(k,c,T)$.
Maintaining a $\lin(k,c,T)$ using the straightforward approach takes $O(kT)$ words of space and has $O(T)$ update time.
When $k$ is small, the corresponding matrix $A$ is dense.
It is worth noting that the celebrated \Count{} sketch~\cite{CCF04} and the \CountM{} sketch~\cite{CM05} both have this form.\footnote{The \Count{} sketch assigns $c$-wise independent $\pm 1$ random weights to each coordinate, and maintains the weighted sum. One may view all $+1$ coordinates forming one bucket and the $-1$ coordinates forming another bucket. Then the query algorithm may manually subtract the sums of the two buckets to obtain the weighted sum.}
These two linear sketches have wide applications to many streaming problems, including $\ell_2$ norm estimation, $\ell_1$ heavy hitters, $\ell_2$ heavy hitters, point query with $\ell_1$ error, point query with $\ell_2$ error, etc.
%\footnote{$\ell_1$ heavy hitters and point query with $\ell_1$ error require $x$ to be nonnegative at the query time.}

In this paper, we present a direct improvement on the update time of $\lin(k,c,T)$ for small $k$, without using more space.
In particular, when $k<\log^{o(1)} n$ and $T=\Theta(\log n)$, our algorithm has the following guarantees.

\newcommand{\thmmaincont}{
	For any problem that admits a $\lin(k,c,T)$ linear sketch solution, where $k<\log^{o(1)} n$ is a power of two, and $T=O(\log n)$, there is an algorithm that 
	\begin{itemize}
		\item uses $O(k\log n)$ words of space,
		\item has \emph{worst-case} update time $O(\log^{0.582} n)$, and
		\item additive \emph{extra} query time $O(\log^{1.582} n)$
	\end{itemize}
	on a standard word RAM with word-size $w=\Theta(\log n)$, where the exponent $0.582\approx 2\omega/3-1$, and $\omega$ is the current matrix multiplication exponent.
	Moreover, if we allow arbitrary $w$-bit word operations, the update time and extra query time can be further reduced to $O(\log^{\omega/2-1+o(1)} n)=O(\log^{0.187} n)$ and $O(\log^{\omega/2+o(1)} n)=O(\log^{1.187} n)$ respectively.
}

\newcommand{\thmmaincontMM}{
	For any problem that admits a $\lin(k,c,T)$ linear sketch solution, where $k<\log^{o(1)} n$ is a power of two, and $T=O(\log n)$, there is an algorithm that 
	\begin{itemize}
		\item uses $O(k\log n)$ words of space,
		\item has \emph{worst-case} update time $O(\log^{0.582} n)$, and
		\item additive \emph{extra} query time $O(\log^{1.582} n)$
	\end{itemize}
	on a standard word RAM with word-size $w=\Theta(\log n)$, where the exponent $0.582\approx 2\omega/3-1$, and $\omega$ is the current matrix multiplication exponent.
	Moreover, the algorithm can be implemented on a \WRMM{} with the update time and extra query time $O(\log^{\omega/2-1+o(1)} n)=O(\log^{0.187} n)$ and $O(\log^{\omega/2+o(1)} n)=O(\log^{1.187} n)$ respectively.
}

\begin{theorem}\label{thm_main}
	\thmmaincont
\end{theorem}

As an application, our theorem implies that
\begin{itemize}
 	\item one can maintain a constant approximation of the $\ell_2$ norm of $\nu$ with probability $1-1/n^{O(1)}$ using $O(\log n)$ words of space and worst-case update time $O(\log^{0.582} n)$;
 	\item one can construct a data structure for nonnegative $\nu$ using $O(\log n)$ words of space and worst-case update time $O(\log^{0.582} n)$, supporting point queries with $n^{-O(1)}$ failure probability and $\ell_1$-error guarantees, i.e., given an index $i\in[n]$, the data structure returns $\nu_i\pm 0.1 \|\nu\|_1$ with probability $1-1/n^{O(1)}$.
\end{itemize}
Assuming one can do arbitrary word operations, the update times are further reduced to $O(\log^{0.187} n)$.
It is worth noting that the lower bound by Larsen, Nelson and Nguy{\^{e}}n asserts that any nonadaptive \emph{cell-probe} data structure (even if we only count the number of memory accesses at the update time) for the above two problems must have update time at least $\tilde{\Omega}(\sqrt{\log n})$ if one uses $\mathrm{poly}\log n$ space.
Hence, our new streaming algorithm ``breaks'' their lower bound by using adaptivity in the update algorithm.

It turns out the only non-standard $w$-bit word operation needed to get $O(\log^{0.187} n)$ update time is to multiply two $w^{1/2}$ by $w^{1/2}$ 0-1 matrices.
We refer to a word RAM equipped with this matrix multiplication operation as a \WRMM{} (see Section~\ref{sec_wrmm}).

\bigskip

Our new algorithm improves the update time by implicitly maintaining the linear sketch.
Hence, in order to answer a query, one would need to first preprocess the data structure to recover the sketches.
It may cause a slower query time by having a preprocessing stage with $O(\log^{1.582} n)$ time.
Some applications may have a very slow query time to begin with (e.g., $\ell_1$ heavy hitters~\cite{CM05}), in which case, the additive $O(\log^{1.582} n)$ time is insignificant.
However, for applications with $O(\log n)$ query time (e.g., point query with $\ell_1$ error~\cite{CM05}), the preprocessing time becomes the bottleneck.

\begin{figure}[t]
\centering
\includegraphics[width=12cm]{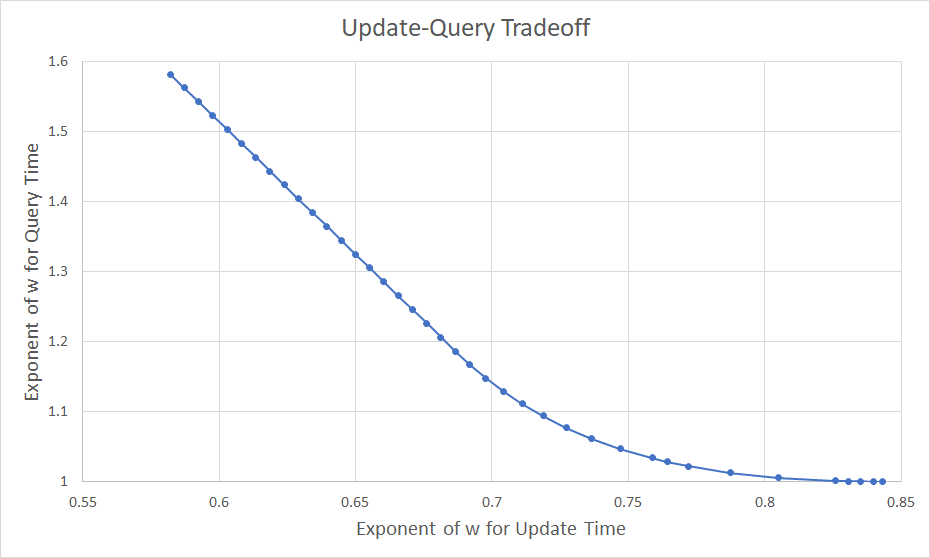}
\caption{Trade-off between the exponent of the query time and update time in our algorithm in the word RAM model, from combining Theorems~\ref{thm_main} and \ref{thm:mm} below with the best known algorithm for rectangular matrix multiplication~\cite{legallrect}.}
\label{fig_tradeoff}
\end{figure}

To allow for a faster query time for those problems, we further provide a smooth tradeoff between the update time and the extra query time (see Figure~~\ref{fig_tradeoff}).
In particular, the other end of this tradeoff curve is an algorithm with $\log^{0.844} n$ update time and $\log^{1+o(1)} n$ extra query time.
To the best of our knowledge, in every application, the query algorithm needs to spend at least a constant time on each of the $O(\log n)$ sketches, the query time was already at least $\Omega(\log n)$.
This tradeoff point almost does not slow down the query time, while it still improves the update time non-trivially.
\newcommand{\thmlowqcont}{
	For any problem that admits a $\lin(k,c,T)$ linear sketch solution, where $k<\log^{o(1)} n$ is a power of two, and $T=O(\log n)$, there is an algorithm that 
	\begin{itemize}
		\item uses $O(k\log n)$ words of space,
		\item has \emph{worst-case} update time $O(\log^{0.844} n)$, and
		\item additive \emph{extra} query time $O(\log^{1+o(1)} n)$
	\end{itemize}
	on a standard word RAM with word-size $w=\Theta(\log n)$, where the exponent $0.844\approx 1-\alpha/2$, and $\alpha$ is the current dual matrix multiplication exponent.
}
\begin{theorem}\label{thm_lowq}
	\thmlowqcont
\end{theorem}

To obtain the faster update time, our new algorithms use the following simple observation that connects the \emph{worst-case} update time to the running time of a computation problem.
\begin{proposition}[informal]\label{prop_batch}
	For any streaming algorithm $\cA$ that uses $O(S)$ words of space, if one can perform $O(S)$ updates in $O(t_{\batch})$ time and $O(S)$ words of (extra) space, then $\cA$ can be simulated using $O(S)$ words and $O(t_{\batch}/S)$ \emph{worst-case} update time.
\end{proposition}
As a byproduct, we obtain a ``hardness result for proving hardness.''
The above proposition implies that proving any super constant update time lower bound, even assuming that the data structure uses optimal space up to a constant factor, would result in a super linear time lower bound for linear space algorithms solving some computation problem in the RAM model, which seems beyond the scope of our current techniques.
In particular, Theorem~\ref{thm_main} implies that proving any $\log^{\epsilon} n$ lower bound for any problem that admits a $\lin(O(1),O(1),\log n)$ solution would imply a non-trivial lower bound for matrix multiplication.
We note that although all computational problems admit linear time algorithm in the \emph{cell-probe} model, Proposition~\ref{prop_batch} \emph{does not} imply a streaming algorithm with faster update time in the \emph{cell-probe} model.
We also note that recently, Dvir, Golovnev and Weinstein~\cite{DGW19} proved a hardness result for proving space-\emph{query time} tradeoff lower bounds, but with less ``severe'' consequences.
See more discussions in the next subsection.

\subsection{Our technique}\label{subsec:ourtechnique}
The key steps that lead to the improved update time are
\begin{itemize}
	\item observing that handling a batch of $B$ updates in $o(B \cdot T)$ time implies an $o(T)$ \emph{worst-case} update time algorithm;
	\item designing a faster batch update algorithm for linear sketches.
\end{itemize}
As a running example, let us focus on $\lin(2,2,\log n)$, i.e., we partition $[n]$ into two buckets using pairwise hash families, maintain the sum of each bucket, and repeat $T=\log n$ times independently.
Note that the straightforward implementation takes $O(\log n)$ space and update time.

Handling $B$ updates in $o(BT)=o(B\log n)$ time naturally implies an algorithm with $o(\log n)$ \emph{amortized} update time: temporarily store each update in a buffer of size $B$, and handle the updates all at once for every $B$ updates. 
By applying a trick similar to \emph{global rebuilding} of Overmars~\cite{Over83}, we show that a faster batch update algorithm also leads to lower \emph{worst-case} update time.
To this end, we store two buffers each of size $B$, one of which is the \emph{active} buffer.
We assume that the first buffer is active in the beginning.
The new updates are always appended to the active buffer.
Then after $B$ updates, the active buffer becomes full.
Now we switch the second buffer to be active, and in the next $B$ updates, we fill up the second buffer while \emph{gradually flushing the first buffer}.
That is, if there is an $O(B t_u)$ time algorithm handling all $B$ updates, we are going to simulate it for $O(t_u)$ steps in each of the next $B$ updates.
Therefore, when the second buffer becomes full again, we will have already emptied the first buffer, and now we can switch the first buffer back to active and start to flush the second buffer.
Each update takes $O(t_u)$ time in worst case.
We present the details in Section~\ref{sec_reduction}.
Note that this reduction \emph{does not} work for cell-probe algorithms, because it requires us to simulate the algorithm for a small number of steps in each update.
Such simulation is only possible on a RAM for a RAM algorithm, but not possible in the cell-probe model for a cell-probe algorithm.

\bigskip

For $\lin(2,2,\log n)$, the space usage is $O(\log n)$ words.
Hence, we can afford to set $B=\log n$ without using more space (except for a constant factor). 
We then show that $\log n$ updates can be applied all together in $o(\log^2 n)$ time.
Each update $(u, \Delta)$ adds the $u$-th column of $A$ multiplied by $\Delta$, to $A\nu$.
Thus, $\log n$ updates add the sum of $\log n$ columns multiplied by (possibly different) numbers.
To compute this sum, it suffices to multiply a $\log n$ by $\log n$ 0-1 matrix, where each column is a column in $A$, by a $\log n$-dimensional vector with $\log n$-bit entries, where each entry is a $\Delta$.

The task is boiled down to the following: a) given $\log n$ indices $u\in [n]$, compute for each $u$, the $u$-th column in $A$, which reduces to evaluating $\log n$ pairwise independent hash functions for each $u$; b) compute the above matrix-vector product.
We show that for a carefully chosen pairwise hash family, both parts can be reduced to $\log n\times \log n\times \log n$ 0-1 matrix multiplications.

Finally, we show that this matrix multiplication can be computed in $\log^{2\omega/3} n$ time on a word RAM with word-size $w=\log n$, and $\log^{\omega/2} n$ time on a \WRMM{}, where $\omega<2.373$ is the current matrix multiplication exponent, using $O(\log n)$ words of space.\footnote{Note that the input and output sizes are both $O(\log n)$ words.} We do this by combining the usual recursive approach to designing fast matrix multiplication algorithms, together with larger-than-usual base cases which can be multiplied in constant time using word operations. In the \WR~model, this involves packing many short vectors into two words so that, when the words are multiplied as integers, the pairwise inner products between those vectors can be read off from the result.

\bigskip

The extra query time of the above algorithm is $O(\log^{2\omega/3} n)=O(\log^{1.582} n)$ time, since we will have to complete the buffer-flushing algorithm, before the actual query algorithm can be launched.
To reduce the extra query time, we reduce the buffer size, so that flushing the buffer takes only $\log^{1+o(1)} n$ time.
The above arguments still apply, but now, the problem is reduced to rectangular matrix multiplication.
The buffer size is reduced sufficiently so that the three dimensions in the matrix multiplication problem are very imbalanced.
Finally, we present an algorithm that multiplies $w\times w^{\theta}\times w$ 0-1 matrices in $w^{1+o(1)}$ time for some constant $\theta>0$, which is almost linear in the output size.

\subsection{Organization}
In Section~\ref{sec_reduction}, we present the connection between computing batch update efficiently and faster worst-case update time.
In Section~\ref{sec_upper}, we present how to do the batch update using matrix multiplication.
In Section~\ref{sec_mm}, we present the matrix multiplication algorithms for small matrices.

\section{Preliminaries}\label{sec_prelim}

\subsection{Notation}

Throughout this paper, we write $\tilde{O}$ to hide $\poly \log \log(n)$ multiplicative factors, so $\tilde{O}(T) = T \cdot \poly \log \log n$ and $\tilde{O}(1) = \poly\log\log n$. When $q$ is a power of a prime, we write $\F_q$ for the finite field of order $q$. For $n \in \N$, we write $[n] := \{1,2,\ldots,n\}$.

\subsection{Model of Computation}

We focus in this paper on the word RAM model of computation~\cite{fredman1990blasting}. In the model for word size $w$ (typically we pick $w = \Theta(\log n)$ where $n$ is the input size), the algorithm has random access to words of memory, each of which stores $w$ bits. An algorithm is allowed to perform any ``standard'' word operations, which only take as input a constant number of words, in constant time. Of course, the efficiency of an algorithm can vary depending on what word operations are considered standard; here we consider two options.

\subsubsection{Standard \WR}

In the first option, which we will call the \WR~model, we only allow for the following ``simple'' word operations: $+,-,\times,/$, bit-wise AND, OR, XOR, negation, and bit-shifts.\footnote{In fact, multiplication can be replaced with just bit-shifts, and multiplication can still be performed in $\tilde{O}(1)$ time~\cite{brodnik1997trans}.} This is the most basic definition of the \WR~model used in the literature, and as these operations are so standard, that algorithms designed in this model should be implementable in any word RAM architecture. Nonetheless, many more complicated operations are known to require only $\tilde{O}(1)$ time as well, by combining the simple operations in clever ways. We will make use of the following operations from past work:

\begin{proposition}[{\cite[Proposition~1]{brodnik1997trans}}]\label{prop:permute}
For any fixed permutation $\pi$ on $m$ symbols, given a bit vector $x[1]x[2]\cdots x[m]$, we can compute the permutation $x[\pi(1)]x[\pi(2)]\cdots x[\pi(m)]$ in time $\tilde{O}(\lceil m/w \rceil)$.
\end{proposition}

\begin{proposition}[{\cite[Theorem~8.7]{aho1974design}}]\label{prop:polyarith}
Given two polynomials $f,g \in \F_2[z]$ of degree at most $w$, we can compute $f(z) \cdot g(z)$ and $f(z) \pmod{g(z)}$ in time $\tilde{O}(1)$.
\end{proposition}

Therefore, if we represent each field element in $\F_{2^w}$ as a degree-$w$ polynomial $f(z)$, and encode it by writing down the coefficients, then any field operation can be computed in $\tilde{O}(1)$ time.

\subsubsection{Word~RAM$^{\text{MM}}$}\label{sec_wrmm}

In the second option, which we call the \WRMM~model, an algorithm is additionally allowed to multiply any two matrices over $\Z$ which each fit into a single word as a single word operation. For instance, the algorithm could multiply two $\sqrt{w} \times \sqrt{w}$ matrices whose entires are $O(1)$-bit integers in $\tilde{O}(1)$ time. This model is not necessarily unrealistic: in practice, engineers may design specific hardware to handle certain important word operations in order to speed up algorithms, and in particular, there has been substantial work on hardware for the fast multiplication of small matrices (see e.g.~\cite{fatahalian2004understanding,fujimoto2008faster,mnih2009cudamat,lee2010debunking}).

That said, this model is particularly interesting from a theoretical perspective, because of its relationship with the \emph{cell-probe model}. Lower bounds against streaming algorithms are typically proved in the cell-probe model, in which an algorithm is only charged for accessing words of memory, and not for any computation on the contents; the \WRMM~model is still weaker than the cell-probe model. We will show that a $\lin(2,2,\log n)$ linear sketch can be maintained with $O(w^{\omega/2 - 1})$ update time and $O(w^{\omega/2})$ query time in the \WRMM~model, where $\omega$ is the matrix multiplication constant\footnote{Here we are using the `usual' definition of $\omega$, in terms of the size of an arithmetic circuit for performing matrix multiplication. Matrix multiplication can be performed faster than this in the cell probe model, but as discussed in Section~\ref{subsec:ourtechnique}, it is not evident how to use this in conjunction with our approach to design a faster cell probe algorithm.}. If $\omega=2$ then this would result in $\tilde{O}(1)$ update time and $\tilde{O}(w)$ query time. Hence, any cell-probe lower bounds for this problem could have substantial implications in arithmetic complexity theory.

\section{Reduction to Batch Problem}\label{sec_reduction}

Consider a dynamic problem $\cP$ with updates $\cU$ (and queries $\cQ$) on word RAM with word-size $w$.
For now, let us assume each update can fit in one word, i.e., $|\cU|\leq 2^w$.
Let $S$ be the minimum number of \emph{words} required to solve $\cP$.
In the following, we show that an algorithm that handles a batch of updates can be transformed into a data structure with \emph{worst-case} update time.

More specifically, fix a data structure $\cD$ for $\cP$ that uses $S$ words of space and has query time $t_q$, and consider the following computational problem.
\begin{center}
\parbox{0.9\textwidth}{
	{\bf \bud{B}:} Given a memory state $M$ of $\cD$ and $B$ updates $u_1,u_2,\ldots,u_B\in\cU$, compute the new memory state after all $B$ updates are applied in the order.
	If $\cD$ is randomized, sample a new memory state according to the distribution defined by $\cD$.
}
\end{center}
Note that the input length of this problem is $O(S+B)$ words.
The following theorem asserts that if the $B$ updates can be handled in a batch efficiently, then $\cP$ has a solution with \emph{worst-case} update time.

\begin{restate}[Proposition~\ref{prop_batch}]
	For any $B\leq O(S)$, if there is a RAM algorithm $\cA$ that solves \bud{B} in time $t_{\batch}$ and space $O(S)$ words, then there is a data structure for $\cP$ that
	\begin{itemize}
	 	\item uses $O(S)$ words of space,
	 	\item has worst-case update time $O(t_{\batch}/B)$, and
	 	\item worst-case query time $O(t_{\batch}+t_q)$.
	 \end{itemize}
\end{restate}
In particular, the contrapositive implies if the update time must be super constant, then \bud{B} has no linear time algorithm on RAM.

\begin{proof}
To construct a data structure with fast update time, the natural idea is to buffer the updates and handle the updates in a batch using $\cA$.
The update time would then be $O(t_{\batch}/B)$ \emph{amortized}.
However, it turns out that a simple trick can deamortizes it.

We will use two buffers \buf$_0$ and \buf$_1$, both of size $B$.
Each time we receive an update, it is put into \buf$_0$.
Once \buf$_0$ becomes full, we are going to put the subsequent updates into \buf$_1$, while at the same time we gradually flush the first buffer \buf$_0$.
That is, each time we receive a new update, it is put into \buf$_1$, then we simulate $\cA$, which handles all $B$ updates in \buf$_0$ in $t_{\batch}$ time, for $O(t_{\batch}/B)$ steps.\footnote{Note that a RAM algorithm uses only $O(1)$ registers, including a pointer to the line of code it is currently executing. Hence, with an extra counter, one can run a RAM algorithm for a certain number of steps and pause. Next time, we may continue from there.}
Since \buf$_1$ can hold another $B$ updates, we will be able to finish simulating $\cA$ before it gets full.
Once \buf$_1$ becomes full, we will switch the roles of the two buffers: put the subsequent updates into \buf$_0$ and gradually flush \buf$_1$ using $\cA$.
To answer a query, it suffices to finish flushing the buffer, and then run the query algorithm of $\cD$.
See Figure~\ref{fig_datastr} for details.

Since $\cA$ uses $O(B)$ words of space, the total space usage is $O(B)$.
The update time is $O(t_{\batch}/B)$ in worst-case, and the query time is $O(t_{\batch}+t_q)$.

\begin{figure}[!ht]
\begin{center}
\fbox{
{\footnotesize

\parbox{6.375in} {
{memory}:
\vspace{-4pt}\begin{enumerate}
	\addtolength{\itemsep}{-4pt}
	\item $M$: memory state of $\cD$, initialized according to $\cD$ \hfill // $S$ words
	\item \buf$_0$, \buf$_1$: two initially empty buffers that each can store up to $B$ updates \hfill // $B$ words each
	\item $b$: indicate the active buffer, initially set to $0$ \hfill // $1$ bit
	\item {\it temp}: the working memory of $\cA$ \hfill // $O(S)$ words
\end{enumerate}

\medskip
{update}($u$): \hfill // handle an update $u\in \cU$
\vspace{-4pt}\begin{enumerate}
\addtolength{\itemsep}{-4pt}
\item append $u$ to \buf$_b$
\item {\bf if} \buf$_b$ is full {\bf then}
\item \quad $b\leftarrow 1-b$
\item \quad start a new instance of {background\_update}()
\item {\bf if} {background\_update}() has not terminated {\bf then}
\item \quad simulate {background\_update}() for $O(t_{\batch}/S)$ steps
\end{enumerate} 

background\_update(): 
\vspace{-4pt}\begin{enumerate}
\addtolength{\itemsep}{-4pt}
	\item reset {\it temp}
	\item run $\cA$ on ($M$, \buf$_{1-b}$) and obtain new memory state $M'$
	\item copy $M'$ to $M$
	\item empty \buf$_{1-b}$
\end{enumerate}

\bigskip
query($q$): \hfill // handle a query $q\in \cQ$
\vspace{-4pt}\begin{enumerate}
\addtolength{\itemsep}{-4pt}
	\item {\bf if} {background\_update}() has not terminated {\bf then}
	\item \quad finish {background\_update}()
	\item $b\leftarrow 1-b$
	\item background\_update()
	\item run the query algorithm of $\cD$
\end{enumerate}
}}}
\caption{Data structure for $\cP$ with $O(t_{\batch}/B)$ worst-case update time and $O(S)$ space.}\label{fig_datastr}
\end{center}
\end{figure}

\end{proof}

\section{Update Efficient Streaming Algorithm}\label{sec_upper}

In this section, we present our update-efficient streaming algorithm for $\lin(k,c,T)$, and prove Theorem~\ref{thm_main}.

\begin{restate}[Theorem~\ref{thm_main}]
	\thmmaincontMM
\end{restate}

To prove the theorem, we first apply Proposition~\ref{prop_batch} and set the buffer size $B=\log n$, and reduce the problem to applying $\log n$ updates in batch. 
It turns out that to apply $\log n$ updates, it suffices to 
\begin{enumerate}
	\item evaluate $\log n$ $c$-wise independent hash families on all $\log n$ updated indices, and
	\item compute a matrix-vector product, where the matrix is a $\log n\times \log n$ binary matrix, and the vector has $\log n$ dimensions with $\log n$-bit numbers in the entries.
\end{enumerate}

In the subsections below, we show that both tasks can be done efficiently using fast matrix multiplication for small matrices.
\newcommand{\lemhashcont}{
	For any constant $c\geq 2$, there is a $c$-wise independent hash family $h_s:\{0,1\}^w\rightarrow \{0,1\}$ for $s\in \{0,1\}^{cw}$, such that given $w$ seeds $s_1,\ldots,s_w$ and $w$ inputs $u_1,\ldots,u_w$, one can compute $w$ $w$-bit strings $v_1,\ldots,v_w$ in $O(w^{1.582})$ time and $O(w)$ words of space, such that $v_i$ stores the $w$ hash values of $u_i$, i.e., the $j$-th bit of $v_i$ is equal to $h_{s_j}(u_i)$ for all $i,j\in[w]$.
	Moreover, the running time can be reduced to $O(w^{1.187})$ on a \WRMM{}.
}
\begin{lemma}\label{lem_hash}
	\lemhashcont
\end{lemma}

\newcommand{\lemmvprodcont}{
	Given as input a matrix $A \in \{0,1\}^{w \times w}$ and a vector $v \in \Z^w$ of $w$-bit integers, one can compute the product $Av$ in $O(w)$ words of space and time $O(w^{1.582})$ on a word RAM, and in time $O(w^{1.187})$ on a \WRMM{}.
}
\begin{lemma}\label{lem_mvprod}
	\lemmvprodcont
\end{lemma}

Now, we proof Theorem~\ref{thm_main} using the above lemmas.

\begin{proof}[Proof of Theorem~\ref{thm_main}]
	We first design a fast batch update algorithm, which handles $w=\log n$ updates to $\lin(k,c,\log n)$ in $O(w^{1.582})$ time on a word RAM.
	The algorithm is given as input, a memory state of $\lin(k,c,\log n)$ and $w$ updates $(u_1,\Delta_1)$, $\ldots$, $(u_w,\Delta_w)$.
	By definition, $\lin(k,c,w)$ uses $w$ $c$-wise independent hash functions $h_1,\ldots,h_w:[n]\rightarrow [k]$, and stores the random seeds, as well as $kw$ counters: $\sum_{u:h_j(u)=b}\nu_u$ for each $j\in[w]$ and $b\in [k]$.
	To perform the batch update, we first compute for all $i,j\in[w]$, which of the $kw$ counters $\Delta_i$ needs to be added to. 

	To this end, we use the $c$-wise independent hash family and the batch evaluation algorithm in Lemma~\ref{lem_hash}.
	Note that Lemma~\ref{lem_hash} only considers such hash functions with one bit output.
	To apply to our problem, we apply the lemma on each of the $\log k$ output bits, since $k$ is a power of two.
	Therefore, in $O(w^{1.582})$ time, we compute for every $u_i$, $\log k$ bit-strings $v'_{i,1},\ldots,v'_{i,\log k}$ such that for each $j\in [w]$, the $j$-th bits of the $\log k$ strings encode the binary representation of $h_j(u_i)$.

	Next, we post-process the binary representations into indicator vectors.
	That is, we will compute a $kw$-bit string $v_i$, stored in $k$ words, that encodes for each of the $kw$ counters, whether $\Delta_i$ needs to be added to it. 
	This can be done in $O(k\log k)$ time: For each $b\in[k]$, a $w$-bit string indicating if $h_j(u_i)=b$ for each $j\in[w]$, can be computed by doing $O(\log k)$ bit-wise ANDs and negations.

	% The following claim asserts that such conversion can be implemented efficiently.
	% To focus on the main algorithm, we defer its proof to the end.
	% \begin{claim}\label{claim_trans}
	% 	Given $\log k$ $w$-bit strings $v'_1,\ldots,v'_{\log k}$, one can compute $w$-bit strings $v_1,\ldots,v_k$ in $O(k\log k)$ time, such that for $j\in [w]$ and $l\in[k]$, the $j$-th bit of $v_l$ is $1$ if and only if the $j$-th bits of $v'_1,\ldots,v'_{\log k}$ is the binary representation of $l$.
	% \end{claim}

	After such transformation, we compute the changes to the $kw$ counters. 
	More specifically, the $l$-th counter needs to increase by $\sum_{\textrm{$i$: the $l$-th bit of $v_i$ $=1$}} \Delta_i$.
	The task is precisely computing the matrix-vector multiplication of 
	\begin{itemize}
		\item a $kw$ by $w$ binary matrix, whose columns are $v_1,\ldots,v_w$, and
		\item a $w$ dimensional vector, whose entries are $w$-bit integers $\Delta_i$.
	\end{itemize}
	Next, we will use Lemma~\ref{lem_mvprod} to compute the product.
	We first use Proposition~\ref{prop:permute} to permute the bits in the matrix and the vector, so that it matches the input format of the lemma, which takes $\tilde{O}(kw)$ time.
	Then, we apply the lemma, since $k<\log^{o(1)} n$, the product can be computed in $O(w^{1.582})$ time.
	At last, we use Proposition~\ref{prop:permute} again to permute the bits in the output, so that the $l$-th word of the output contains the $l$-th entry in the product vector, i.e., the change to the $l$-th counter.
	We add the product to counters.
	
	The batch update algorithm for $w$ updates runs in $O(w^{1.582})$ time, and $O(kw)$ space.
	Finally, the theorem is proved by applying Proposition~\ref{prop_batch} for $B=w$.
	By applying the \WRMM{} versions of Lemma~\ref{lem_hash} and Lemma~\ref{lem_mvprod}, we obtain the claimed update and query times for \WRMM{}.
\end{proof}

% It remains to prove Claim~\ref{claim_trans}.

% \begin{proof}[Proof of Claim~\ref{claim_trans}]
% 	...\mnote{add a proof}
% \end{proof}

	Both of Lemma~\ref{lem_hash} and Lemma~\ref{lem_mvprod} use fast matrix multiplication for small matrices.

	\begin{theorem}\label{thm_mm_mm}
	In the \WRMM~model, one can perform $w \times w \times w$ matrix multiplication over $\F_q$ for $q \leq \poly(w)$ in time $O(w^{\omega/2+\eps})$ for any $\eps>0$.
	\end{theorem}

	\begin{theorem}\label{thm:mm}
	In the \WR~model, for any $p \in [0,1]$, and any $\eps>0$, one can perform $w \times w^p \times w$ matrix multiplication over $\F_q$ for $q \leq \poly(w)$ in time 
	\begin{itemize}
	    \item $O(w^{(1+p)\omega/3 + \eps}) \leq O(w^{0.791 \cdot (1+p)})$ if $p \geq 1/2$,
	    \item $O(w^{\omega(1,2p,1)/2 + \eps})$ if $p \leq 1/2$.
	\end{itemize}
	\end{theorem}

	The proof of the two theorems are deferred to Section~\ref{sec_mm}. 
	In the next two subsections, we prove Lemma~\ref{lem_hash} and Lemma~\ref{lem_mvprod} respectively.

    \subsection{Evaluating $c$-wise Hash Functions}

	In the following, we prove Lemma~\ref{lem_hash}.

	\begin{restate}[Lemma~\ref{lem_hash}]
		\lemhashcont
	\end{restate}

	\begin{proof}
		We use the following $c$-wise independent family $h_s$: given input $u\in\{0,1\}^w$, first generate a vector $g(u)\in \bF_2^{cw}$ such that for any $c$ different $u_1,\ldots,u_c\in\{0,1\}^w$, the corresponding vectors $g(u_1),\ldots,g(u_c)$ are linearly independent; then we take the seed $s$ also in $\bF_2^{cw}$, and let
		\[
			h_s(u):=\left<s,g(u)\right>.
		\]
		Note that $h_s$ is $c$-wise independent, because for any $u_1,\ldots,u_c\in\{0,1\}^w$, and $y_1,\ldots,y_c\in\{0,1\}$, we have
		\[
			\Pr_s[\forall i,\left<s,g(u_i)\right>=y_i]=2^{-c},
		\]
		due to the linear independence of $\{g(u_i)\}$.
		
		Suppose such $g(u)$ can be computed efficiently, then evaluating all $h_{s_j}$ on all $u_i$ becomes to compute a matrix multiplication over $\bF_2$: 
		Let $\mathbf{S}$ be a $w$ by $cw$ matrix, where the $j$-th row is $s_j$ for all $j$, and let $\mathbf{X}$ be a $cw$ by $w$ matrix, where the $i$-th column is $g(u_i)$ for all $i$, then the $(j,i)$-th entry of the product matrix $\mathbf{SX}$ is exactly the inner product $\left<s_j,g(u_i)\right>$, i.e., $h_{s_j}(u_i)$.
		By Theorem~\ref{thm_mm_mm} and Theorem~\ref{thm:mm}, $\mathbf{SX}$ can be computed in $O(w^{1.582})$ time on a word RAM, or $O(w^{1.187})$ time on a \WRMM.

		\bigskip
		Next, we give a construction of $g(u)$ with the above $c$-wise linear independence property, and show that it can be computed efficiently.
		We first view $u\in\{0,1\}^w$ as an element in $\bF_{2^w}$, encoded in a canonical form.
		That is, we fix any irreducible degree-$w$ polynomial $Q(z)\in\bF_2[z]$, and store it explicitly in memory.
		Suppose $u=(u(0),u(1),\ldots,u(w-1))$, we view it as the polynomial $u(0)+u(1)z+\cdots+u(w-1)z^{w-1}$ in $\bF_2[z]$, which is an element in $\bF_2[z]/(Q)\cong\bF_{2^w}$.

		We define $g(u):=(1,u,u^2,\ldots,u^{c-1})$, where each $u^i$ is computed in $\bF_{2^w}$ and uses the above encoding.
		Thus, $g(u)$ can either be viewed as a vector in $\bF_{2^w}^c$ or a vector in $\bF_2^{cw}$.
		Note that in both views, adding two vectors yield the same result (both are bit-wise XOR).
		By the fact that Vandermonde matrices have full rank, for any different $u_1,\ldots,u_c$, and $\alpha_1,\ldots,\alpha_c\in \bF_{2^w}$ that are not all zero, we have 
		\[
			\alpha_1\cdot g(u_1)+\cdots+\alpha_c\cdot g(u_c)\neq 0.
		\]
		In particular, it holds for any $\alpha_1,\ldots,\alpha_c\in \{0,1\}$ that are not all zero, which implies that $g(u_1),\ldots,g(u_c)$ are linearly independent as vectors in $\bF_2^{cw}$.
		Finally, by Proposition~\ref{prop:polyarith}, each $g(u)$ can be computed in $\tilde{O}(c)$ time.
		This proves the lemma.
	\end{proof}
	
	\subsection{Matrix-Vector Multiplication}
	
    In this subsection, we prove Lemma~\ref{lem_mvprod}.

    \begin{restate}[Lemma~\ref{lem_mvprod}]
    	\lemmvprodcont
    \end{restate}

    \begin{proof}
    First, construct the matrix $B \in \{0,1\}^{w \times w}$, whose entry $B[i,j]$ is the $j$th bit of $v[i]$ when written out in binary; in particular, $v[i] = \sum_{j=1}^w 2^{j-1} B[i,j]$. Next, use the given algorithm to compute the product $C := AB$ over $\F_q$. Note that since $q>w$ and the entries of $A$ and $B$ are all in $\{0,1\}$, it follows that $C$ is also the product of $A$ and $B$ over $\Z$. Finally, output the vector $v' \in \Z^w$ given by $v'[i] = \sum_{j=1}^w 2^{j-1} C[i,j]$. To see that it is correct, note that:
    \begin{align*}
    	v'[i] &= \sum_{j=1}^w 2^{j-1} C[i,j] \\
    	&= \sum_{j=1}^w 2^{j-1} \left( \sum_{k=1}^w A[i,k] B[k,j] \right) \\
    	&= \sum_{k=1}^w A[i,k] \left( \sum_{j=1}^w 2^{j-1} B[k,j] \right) \\
    	&= \sum_{k=1}^w A[i,k] v[k],
    \end{align*}
    which is exactly the desired output.
    The bottleneck of the running time is to compute the matrix product $AB$, which by Theorem~\ref{thm_mm_mm} and Theorem~\ref{thm:mm}, has the claimed running time.
    \end{proof}

\subsection{Faster query time}
	In this subsection, we describe how to obtain $\log^{1+o(1)} n$ extra query time.
	Again, we use Proposition~\ref{prop_batch}.
	The idea is to set the buffer size $B$ to be smaller, so that $B$ updates can be handled in $\log^{1+o(1)} n$ time.
	Hence, the worst-case update time is $(\log^{1+o(1)} n)/B$.
	To this end, let $\theta$ be a positive number such that $w\times w^{\theta}\times w$ matrix multiplication can be computed in $w^{1+o(1)}$ time.
	By Theorem~\ref{thm:mm}, we can set $\theta>0.156$, according to the current best rectangular matrix multiplication algorithm.
	We set $B=\log^{\theta} n$.

	A similar argument to the proof of Theorem~\ref{thm_main} shows that the problem reduces to evaluating $\log n$ $c$-wise hash functions on $\log^{\theta} n$ points, as well as computing a matrix-vector product, where the matrix is a 0-1 matrix of size $(k\log n)\times \log^{\theta} n$, and the vector has dimension $\log^{\theta} n$ and $(\log n)$-bit values.
	Finally, similar to the proofs of Lemma~\ref{lem_hash} and Lemma~\ref{lem_mvprod}, both problems can be reduced to computing $\log n\times \log^{\theta} n\times \log n$ matrix multiplication, which takes $\log^{1+o(1)} n$ by the definition of $\theta$.
	This gives us an algorithm with update time $O(\log^{0.844} n)$, and extra query time $\log^{1+o(1)} n$, proving Theorem~\ref{thm_lowq}.
	We omit the rest of the details.

	\begin{restate}[Theorem~\ref{thm_lowq}]
		\thmlowqcont
	\end{restate}

\section{Fast Matrix Multiplication}\label{sec_mm}

\subsection{Tensor Rank and Matrix Multiplication}\label{subsec:tensors}

In this section, we show how to take advantage of the word RAM model to speed up matrix multiplication when the dimensions of the matrices are polynomials in the word size $w$. We begin by reviewing useful notation related to fast matrix multiplication algorithms.

Let $\F$ be any field, $a,b,c$ be any nonnegative real numbers, $n$ be any positive integer, and $X = \{ x_{i,j} \}_{i\in [n^a], j \in [n^b]}$, $Y = \{ y_{j,k} \}_{j\in[n^b], k \in [n^c]}$, and $Z = \{ z_{i,k} \}_{i\in [n^a], k \in [n^c]}$ be three sets of formal variables. The \emph{rank of $n^a \times n^b \times n^c$ matrix multiplication over $\F$}, denoted $R_\F(\langle n^a, n^b, n^c \rangle)$, is the smallest integer $r$ such that there are values $\alpha_{i,j,\ell}, \beta_{j,k,\ell}, \gamma_{i,k,\ell} \in \F$ for all $i \in [n^a]$, $j \in [n^b]$, $k \in [n^c]$ and $\ell \in [r]$ such that
\begin{align}\label{rk}\sum_{\ell = 1}^r \left( \sum_{i \in [n^a], j \in [n^b]} \alpha_{i,j,\ell} x_{i,j}\right) \left( \sum_{j \in [n^b], k \in [n^c]} \beta_{j,k,\ell} y_{j,k}\right) \left( \sum_{i \in [n^a], k \in [n^c]} \gamma_{i,k,\ell} z_{i,k}\right) = \sum_{i \in [n^a], j \in [n^b], k \in [n^c]} x_{i,j} y_{j,k} z_{i,k}.\end{align}

\begin{proposition}\cite[Proposition 15.1]{BCSbook} \label{rkfromomega}
For any positive real $a,b,c$ and field $\F$, suppose there is a $t > 0$ and an algorithm, in the arithmetic circuit model, which performs $n^a \times n^b \times n^c$ matrix multiplication over the field $\F$ using $n^{t+o(1)}$ field operations. Then, for every $\eps>0$, there is a positive integer $q$ such that $R_\F (\langle q^a, q^b, q^c \rangle) \leq q^{t+\eps}$.
\end{proposition}

We define $\omega_\F(a,b,c) := \liminf_{q \in \N} \log_q(R_\F(\langle q^a, q^b, q^c \rangle))$. It is known (and we will show below in Theorem~\ref{mmalg}) that $n^a \times n^b \times n^c$ matrix multiplication over $\F$ can be performed in $O(n^{\omega_\F(a,b,c) + \eps})$ field operations. Although $\omega_\F(a,b,c)$ may differ depending on the field $\F$, all known constructions achieve the same value for all $\F$, so we will typically drop the $\F$ and simply write $\omega(a,b,c)$ as in past work. We also write $\omega = \omega(1,1,1)$. We note a couple of simple properties:
\begin{itemize}
    \item for all $a,b,c,d \geq 0$ we have $\omega(d \cdot a,d \cdot b,d \cdot c) = d \cdot \omega(a,b,c)$.
    \item $\omega(a,b,c) = \omega(b,c,a)$ (or more generally any permutation of the three arguments) by the symmetry of the right-hand side of (\ref{rk}).
\end{itemize}

\subsection{New algorithms for small matrices}

We now show how to design faster algorithms for multiplying small matrices, whose dimensions are polynomials in the word size $w$ of the word RAM model. Our algorithm only slightly modifies the usual recursive algorithm for fast matrix multiplication by making use of a more efficient base case.

We state our result over the field $\F_p$ for $p \leq \poly(w)$, but it generalizes to any field where operations can be performed efficiently in the word RAM model.

\begin{theorem} \label{mmalg}
Let $p \leq \poly(w)$ be any prime number. Suppose, for some nonnegative real numbers $a',b',c'$, that there is an algorithm which performs $w^{a'} \times w^{b'} \times w^{c'}$ matrix multiplication over the field $\F_p$ in time $M(w)$. Then, for any nonnegative real numbers $a,b,c$, and any $\eps>0$, there is an algorithm which performs $w^{a+a'} \times w^{b+b'} \times w^{c+c'}$ matrix multiplication over $\F_q$ in time $O(M(w) \cdot w^{\omega(a,b,c) + \eps})$.
\end{theorem}

\begin{proof}
We design a recursive algorithm which, for all positive integers $n$, performs $(n^a w^{a'}) \times (n^b w^{b'}) \times (n^c w^{c'})$ matrix multiplication over $\F_p$ in time $O(M(w) \cdot n^{\omega(a,b,c) + \eps})$. As the base case, when $n=1$, such an algorithm is assumed to exist.

For the recursive step, let $q$ be the positive integer (constant) which is guaranteed to exist by Proposition~\ref{rkfromomega} such that $R(\langle q^a, q^b, q^c \rangle) \leq q^{\omega(a,b,c) + \eps} =: r$, and using the notation of subsection~\ref{subsec:tensors}, let $\alpha_{i,j,\ell}, \beta_{j,k,\ell}, \gamma_{i,k,\ell} \in \F$ for $i \in [q^a]$, $j \in [q^b]$, $k \in [q^c]$, $\ell \in [r]$ be the corresponding coefficients in the rank expression.

Let $A$ be the input matrix of dimensions $w^{a'}n^a \times w^{b'}n^b$ over $\F$, and $B$ be the input matrix of dimensions $w^{b'}n^b \times w^{c'}n^c$ over $\F$.
First, we partition $A$ into a $q^a \times q^b$ block matrix, where each block is a $w^{a'}(n/q)^a \times w^{b'}(n/q)^b$ matrix; call the blocks $A_{i,j}$ for $i \in [q^a]$, $j \in [q^b]$. Similarly we partition $B$ into a $q^b \times q^c$ block matrix, where each block is a $w^{b'}(n/q)^b \times w^{c'}(n/q)^c$ matrix; call the blocks $B_{j,k}$ for $j \in [q^b]$, $k \in [q^c]$. The algorithm first computes, for each $\ell \in [r]$, the linear combination
$$A'_\ell = \sum_{i\in[q^a], j \in [q^b]} \alpha_{i,j,\ell} A_{i,j},$$
and the linear combination
$$B'_\ell = \sum_{j \in [q^b], k \in [q^c]} \beta_{j,k,\ell} B_{j,k}.$$
Since $q$ is a constant, this takes $\tilde{O}(n^{a+b} / w^{1-a-b} + n^{b+c} / w^{1-b-c})$ field operations. (More details here??)

Next, for each $\ell \in [r]$, the algorithm computes the $w^{a'}(n/k)^a \times w^{c'}(n/k)^c$ matrix $C'_\ell := A'_\ell \times B'_\ell$, by \emph{recursively} performing $w^{a'}(n/k)^a \times w^{b'}(n/k)^b \times w^{c'}(n/k)^c$ matrix multiplication. By the inductive hypothesis, this requires $O(r \cdot M(w) \cdot (n/q)^{\omega(a,b,c) + \eps}) = O(M(w) \cdot n^{\omega(a,b,c) + \eps})$ time.

Finally, for each $i \in [q^a]$ and $k \in [q^c]$, the algorithm computes the linear combination
$$C_{i,k} = \sum_{\ell=1}^r \gamma_{j,\ell} C'_j,$$
in total time $O(n^{a+c}/w^{1-a-c})$. These are the blocks of the $w^{a'}n^a \times w^{c'}n^c$ matrix $C$ which we output. We can see these are correct from the definition of the rank expression (equation (\ref{rk}) in subsection~\ref{subsec:tensors}): if we substitute in $A_{i,j}$ for $x_{i,j}$ and $B_{j,k}$ for $y_{j,k}$ in (\ref{rk}), then from the left hand side of (\ref{rk}) we see that $C_{i,k}$ is the resulting coefficient of $z_{i,k}$, and from the right hand size of (\ref{rk}) we see that that coefficient is indeed $\sum_{j} A_{i,j} B_{j,k}$, which is the correct $i,k$ block of the output matrix $C$.

To see that the $O(M(w) \cdot n^{\omega(a,b,c) + \eps})$ running time for the recursive step dominates the other terms $O(n^{a+b}/w^{1-a-b})$, $O(n^{b+c}/w^{1-b-c})$ and $O(n^{a+c}/w^{1-a-c})$, simply note that, because of the time to read the input, we have $n^{\omega(a,b,c)} \geq \Omega(n^{a+b} + n^{b+c} + n^{a+c})$, and $M(w) \geq \Omega(w^{a+b-1} + w^{b+c-1} + w^{a+c-1})$.
\end{proof}

\subsubsection{Word~RAM$^{\text{MM}}$~model}

We begin with the model of computation where matrices which fit into words can be multiplied in constant time. In particular:

\begin{proposition}\label{prop:allmodel}
In the \WRMM~model, one can perform $w^{1/2} \times w^{1/2} \times w^{1/2}$ matrix multiplication over $\F_q$ for $q \leq \poly(w)$ in time $\tilde{O}(1)$.
\end{proposition}

\begin{proof}
A $w^{1/2} \times w^{1/2}$ matrix fits into $\tilde{O}(1)$ words.
\end{proof}

\begin{restate}[Theorem~\ref{thm_mm_mm}]
In the \WRMM~model, one can perform $w \times w \times w$ matrix multiplication over $\F_q$ for $q \leq \poly(w)$ in time $O(w^{\omega/2+\epsilon})$.
\end{restate}

\begin{proof}
Set $a=b=c=a'=b'=c'=1/2$ in Theorem~\ref{mmalg}, combined with the base case from Proposition~\ref{prop:allmodel}.
\end{proof}

\subsubsection{Word RAM model}

\begin{lemma} \label{lem:smallmm}
In the \WR~model with word size $w$, for any positive integers $d_a, d_b, d_c, q$ such that $d_a \cdot d_b \cdot d_c \cdot \log (q) \leq \tilde{O}(w)$, one can compute $d_a \times d_b \times d_c$ matrix multiplication over $\F_q$ in time $\tilde{O}(1)$.
\end{lemma}

\begin{proof}
For a vector $v \in \F_q^{d_b}$, let $\ell = \lceil \log_2(q) \rceil$, and for $i \in [d_b]$, let $v_i \in \{0,1\}^{\ell}$ be the binary representation of $v[i]$ (the $i$th entry of $v$, with leading zeroes added as necessary). Then, letting $g = \lceil \log_2 (2qd_b) \rceil$, define $s(v) \in \{0,1\}^{gd_b}$ to be the string given by $s(v) = 0^{g-\ell} v_{d_b} 0^{g-\ell} v_{d_b-1} \cdots 0^{g-\ell} v_1$. Hence, $s(v)$ is a space-separated concatenation of the entries of $v$, and moreover, as an integer it is equal to $\sum_{i=1}^{d_b} 2^{g(i-1)} \cdot v[i]$. Similarly define $s^r(v) = 0^{g-\ell} v_1 0^{g-\ell} v_{2} \cdots 0^{g-\ell} v_{d_b}$.

Notice that for vectors $v,w \in \F_q^{d_b}$, if we compute $a_{v,w} := s(v) \cdot s^r(w)$ (as a product over the integers) then
$$a_{v,w} = \left( \sum_{i_0=1}^{d_b} 2^{g(i_0-1)} \cdot v[i_0] \right) \cdot \left( \sum_{i_1=1}^{d_b} 2^{g(d_b-i_1)} \cdot w[i_1] \right) = \sum_{j=0}^{2d_b-2} 2^{gj} \sum_{i_0 - i_1 = j+1-d_b} v[i_0] \cdot w[i_1] = \sum_{j=0}^{2d_b-2} 2^{gj} P_{v,w}(j),$$
where we define $P_{v,w}(j) := \sum_{i_0 - i_1 = j+1-d_b} v[i_0] \cdot w[i_1]$. In particular, we know that $P_{v,w}(j)$, which is a sum of at most $d$ products of two integers between $0$ and $q-1$, fits in $g$ bits. Hence, $a_{v,w}$ is a space-separated list of the $P_{v,w}(j)$ for all $j$. Notice in particular that $P_{v,w}(d_b-1)$, when taken mod $q$, is exactly the inner product $\langle v,w \rangle$. Since $s(v), s^r(w)$ fit in $gd_b$ bits, it follows that $a_{v,w}$ fits in $2gd_b$ bits.

Next, for any vectors $v_1, \ldots, v_{d_a} \in \F_q^{d_b}$ and $w_1, \ldots, w_{d_c} \in \F_q^{d_b}$, consider the two strings $$s(v_1, v_2, \ldots, v_{d_a}) := 0^{gd_b}s(v_1)0^{gd_b}s(v_2)\cdots0^{gd_b}s(v_{d_a}),$$ which has length $m_1 := 2gd_bd_a$, and $$s^r(w_1,w_2, \ldots, w_{d_c}) := 0^{m_1}s^r(w_1)0^{m_1}s^r(w_2)\cdots 0^{m_1} s_r(w_{d_c}),$$ which has length $m_2 := d_c(m_1 + gd_b) = O(gd_a d_b d_c)$. These can be computed using Proposition~\ref{prop:permute} in $\tilde{O}(1)$ time. Similar to before, the string $$a_{(v_1, \ldots, v_{d_a}), (w_1, \ldots, w_{d_c})} := s(v_1, v_2, \ldots, v_{d_a}) \cdot s^r(w_1,w_2, \ldots, w_{d_c})$$ is a string of length $O(gd_a d_b d_c)$ which consists of a space-separated list of $a_{v_i, w_j}$ for all $i \in [d_a]$ and $j \in [d_c]$. As above, from this we can extract the inner product $\langle v_i, w_j \rangle$ for all $i \in [d_a]$ and $j \in [d_c]$, which is exactly the desired matrix product.
\end{proof}

\begin{restate}[Theorem~\ref{thm:mm}]
In the \WR~model, for any $p \in [0,1]$, and any $\eps>0$, one can perform $w \times w^p \times w$ matrix multiplication over $\F_q$ for $q \leq \poly(w)$ in time 
\begin{itemize}
    \item $O(w^{(1+p)\cdot\omega/3 + \eps}) \leq O(w^{0.791 \cdot (1+p)})$ if $p \geq 1/2$,
    \item $O(w^{\omega(1,2p,1)/2 + \eps})$ if $p \leq 1/2$.
\end{itemize}
\end{restate}

\begin{proof}
When $p \geq 1/2$, then applying Lemma~\ref{lem:smallmm} with $d_a=d_c=w^{(2-p)/3}$ and $d_b = w^{(2p-1)/3}$, we know there is an algorithm for $w^{(2-p)/3} \times w^{(2p-1)/3} \times w^{(2-p)/3}$ matrix multiplication over $\F_q$ running in time $\tilde{O}(1)$. 
Combining this with Theorem~\ref{mmalg} with $a=b=c=(1+p)/3$, $a'=c'=(2-p)/3$, and $b' = (2p-1)/3$ yields that $w \times w^p \times w$ matrix multiplication over $\F_q$ can be done in the desired running time.

When $p \leq 1/2$, we instead apply Lemma~\ref{lem:smallmm} with $d_a=d_c=w^{1/2}$ and $d_b = 1$, then Theorem~\ref{mmalg} with $a=c=a'=c'=1/2$, $b=p$, and $b'=0$.
\end{proof}

\begin{remark}
When applying Theorem~\ref{thm:mm} with the best known bounds on $\omega$~\cite{williams2012multiplying,le2014powers,legallrect}, the `$+\eps$' terms in the exponents may be replaced by `$+o(1)$'.
\end{remark}

%Josh TODO: write something about $\eps$ vs $o(1)$ in the exponent for rectangular MM.

\paragraph{Acknowledgments.} The authors would like to thank Jelani Nelson for proposing the problem to us, and we would like to thank Jelani Nelson, Virginia Vassilevska Williams and Ryan Williams for helpful discussions.
\bibliography{refs}
\bibliographystyle{alpha}
\end{document}